\documentclass[a4paper,UKenglish]{lipics-v2016}
\usepackage{todonotes}

\usepackage{microtype}

\newcommand{\Q}{\mathcal{Q}}
\newcommand{\C}{\mathcal{C}}
\newcommand{\T}{\mathcal{T}}
\renewcommand{\O}{\hat{O}}
\newcommand{\R}{\mathcal{R}}
\renewcommand{\L}{\mathcal{L}}
\renewcommand{\l}{\ell}
\newcommand{\leafset}{\mathcal{X}}
\newcommand{\treeord}[2]{T_{#1}(#2)}
\newcommand{\treerev}[2]{\overleftarrow{T}_{#1}( #2 )}
\newcommand{\maxwqtet}{WQC}

\newtheorem{proposition}{Proposition}[section]

\newtheorem{claim}{Claim}

\title{On the Weighted Quartet Consensus problem\footnote{This work was partially supported by the 
Natural Sciences and Engineering Research Council of Canada (NSERC) and the Mitacs Globalink Campus France program.}}
\titlerunning{On the Weighted Quartet Consensus problem} 

\author[1]{Manuel Lafond}
\author[2]{Celine Scornavacca}
\affil[1]{Department of Mathematics and Statistics, University of
Ottawa, 585 King Edward Ave.,  K1N 6N5 Ottawa, Canada \\  \texttt{mlafond2@uottawa.ca}}
\affil[2]{Institut des Sciences de l'Evolution -- Universit\'e Montpellier,  CNRS, IRD, EPHE, Place Eug\`ene Bataillon, 34095 Montpellier, France\\
  \texttt{celine.scornavacca@umontpellier.fr}}
\authorrunning{M. Lafond and C. Scornavacca} 

\Copyright{Manuel Lafond and Celine Scornavacca}

\subjclass{F.2.2 G.2.1 G.2.2}
\keywords{Phylogenetic tree, consensus tree, quartets, complexity, fixed-parameter tractability, approximability.}

\EventEditors{Juha K\"arkk\"ainen, Jakub Radoszewski, and Wojciech Rytter}
\EventNoEds{3}
\EventLongTitle{28th Annual Symposium on Combinatorial Pattern Matching
(CPM 2017)}
\EventShortTitle{CPM 2017}
\EventAcronym{CPM}
\EventYear{2017}
\EventDate{July 4--6, 2017}
\EventLocation{Warsaw, Poland}
\EventLogo{}
\SeriesVolume{78}
\ArticleNo{23}

\newcommand{\new}[1]{#1}
\newcommand{\newml}[1]{#1}

\begin{document}

\maketitle

\begin{abstract}
In phylogenetics, the \emph{consensus} problem consists in summarizing a set of phylogenetic trees that all classify the same set of species into a single tree. Several definitions of consensus exist in the literature; in this paper we focus on the \textsc{Weighted Quartet Consensus} problem, a problem  with unknown complexity status \new{so far}. Here we prove that the \textsc{Weighted Quartet Consensus} problem is NP-hard and we give a $1/2$-factor approximation for this problem.
During the process, we propose a derandomization procedure of a previously known \emph{randomized} $1/3$-factor approximation.
We also investigate the fixed-parameter tractability of this problem. 
\end{abstract}

\section{Introduction}

\emph{Phylogenetics} is the branch of biology that studies evolutionary relationships among different species. 
These relationships are represented via \emph{phylogenetic trees} --  acyclic connected graphs with \new{leaves} labeled by species -- \new{which} are  reconstructed  
from molecular and morphological data \cite{felsenstein2004inferring}.
 One fundamental problem in phylogenetics is to summarize \new{the information contained in} a set of unrooted trees $\mathcal{T}$ classifying \new{the} same set of species into a single tree $T$. 
 The set $\mathcal{T}$  can consist of optimal trees for a single data set, of trees issued from a bootstrap analysis of a unique data set, or even of trees issued from the analysis of different data sets. 
Several consensus methods have been proposed in the past, the \new{probably most known} are the strict consensus  \cite{sokal1981taxonomic,mcmorris1983view} 
and the majority-rule consensus  \cite{margush1981consensusn,barthelemy1986median}.  For a \new{survey}, see \cite{bryant2003classification}. 

In this paper we focus on the \textsc{Weighted Quartet Consensus}  (\maxwqtet) problem \cite{MirarabThesis}, also called the \textsc{Median Tree with Respect to Quartet Distance} problem \cite{bansal2011comparing} and \textsc{Quartet Consensus} problem in 
\cite{jiang2001polynomial}.
Roughly speaking, this problem consists in finding a consensus tree maximizing the weights of the 4-leaf trees -- \emph{quartets} -- it displays, where the weight of a quartet is defined as its frequency in the set of input trees (for a more formal definition, see next section).

More general versions of this problem\new{,} where the input trees  are allowed to have missing species or where the weight of a quartet is not defined w.r.t. a set of trees, are known to be NP-hard \cite{steel1992complexity} (and in fact, even Max-SNP-Hard), but the complexity of the \maxwqtet~problem  \new{has remained open so far}. 
This problem has been conjectured to be NP-hard \cite{bansal2011comparing,MirarabThesis}, and heuristics have recently been implemented and widely used, for instance ASTRAL  \cite{mirarab2014astral}, which is a practical implementation of~Bryant and Steel's work from\cite{bryant2001constructing} (in fact, we show that the ASTRAL algorithm is a $2$-approximation for the minimization version of \maxwqtet).
So far, the best known approximation algorithm for the \maxwqtet~problem consists in choosing a random tree as a solution~\cite{jiang2001polynomial}.
This tree is expected to contain a third of the quartets from the input trees, and so it is a randomized factor $1/3$ approximation.
In \cite{bansal2011comparing}, the \emph{minimization} version of the problem is studied, where the objective is to find a median tree $T$ minimizing the 
sum of quartet \emph{distances} between $T$ and the input trees 
\newml{(the quartet distance 
between two trees $T_1$ and $T_2$ is defined 
as the number of quartets in $T_1$ that are not in $T_2$)}.
A 2-approximation algorithm is given, based on the fact that the quartet distance is a 
metric~\cite{byrka2010new,bansal2011comparing}.

A related problem that has received  more attention is the \emph{Complete Maximum Quartet Compatibility} \new{problem} (CMQC) (see~\cite{berry1999quartet,berry2000practical,jiang2001polynomial,gramm2001minimum,wu2005lookahead, wu2007quartet,chang2010new,morgado2008pseudo,morgado2010combinatorial}).  In the CMQC problem, we are given, for each \new{set $S$ of four species},
exactly one quartet on $S$, and the objective is to find a tree containing a maximum number of quartets from the input.  
This can be seen as \newml{a version} of \maxwqtet~in which
each set of four species has one quartet of weight $1$, and the others have weight $0$.
The CMQC and \maxwqtet~are however fundamentally different.
Although one could apply an algorithm for CMQC to WQC (by keeping only the most frequent quartet on each set of four taxa), 
maximizing the most-frequent quartets may enforce the presence of many low-frequency quartets.
A better solution may prefer more of the middle-frequency quartets, and we give an example of this phenomenon.
It was shown in \cite{jiang2001polynomial} that the CMQC \new{problem} admits a polynomial time approximation scheme (PTAS), but it can only be extended to 
\maxwqtet~intances on a constant number of trees.
Also, CMQC was shown in~\cite{gramm2001minimum,chang2010new} to be fixed-parameter tractable w.r.t. the parameter ``number of quartets to reject from the input''.

The main result of this paper is to establish the NP-hardness of the \maxwqtet~problem.  In Section~\ref{sec:preliminaries}, we introduce preliminary notions, and 
in Section \ref{sec:hardness} we propose a reduction from the NP-hard \textsc{Cyclic Ordering} problem to \maxwqtet.  
It can be shown that this hardness result transfers to the \emph{rooted} setting, in which case we want to optimize \emph{triplets} (3-leaf rooted trees) rather than quartets.
In Section~\ref{sec:nonStructure}, we discuss  how being in a consensus setting, i.e.\ having weights based on a set of input trees on the same leaf set rather than arbitrary weights, does not necessarily make the problem easier, as one could expect:  We list some structural properties \new{that, surprisingly, are not satisfied by optimal solutions of a \maxwqtet~instance}. 
Nevertheless, in Section~\ref{sec:approx} we devise a factor $1/2$ approximation algorithm for \maxwqtet~running in time $O(k^2n^2 + kn^4 + n^5)$, where $k$ is the number of trees and $n$ the number of species (the best known randomized algorithm in the non-consensus setting is a factor $1/3$ one).
As a matter of fact, our algorithm includes a derandomization of this procedure, which had never been done before.
Finally in Section~\ref{FPTAlgos},  we show that the known FPT algorithms for the CMQC problem can be extended to the consensus setting.
This yields an FPT algorithm that is efficient on instances in which there is not too much ambiguity, i.e.\ when few competing quartets on the same $4$ species appear with the same frequency.
We then conclude with some remarks and open problems related to the quartet consensus problem.

\section{Preliminaries}\label{sec:preliminaries}

An \emph{unrooted phylogenetic tree}  $T$ consists of vertices connected by edges, in which \new{every pair of nodes is} connected by exactly one path and no vertex is of degree two. The \emph{leaves} of a tree $T$, denoted by $L(T)$ \new{are} the set of vertices of degree one. Each leaf is associated to a label; the set of labels associated to the leaves of a tree $T$ is denoted by $\L(T)$.  We suppose that there is a bijection between $L(T)$ and  $\L(T)$.  Due to this bijectivity, we will refer to leaves and labels interchangeably.  We denote $|\L(T)|$ as $|T|$. \new{In the following, we will often omit the word ``phylogenetic'' and, unless otherwise stated, all trees are leaf-labeled}. A \emph{binary} unrooted tree has only vertices with degree three and vertices with degree one.  
A rooted (binary) phylogenetic tree is defined in the same way, except that it has exactly one node of degree two called the \emph{root}, denoted by $r(T)$.  Note that sometimes in the literature, rooted trees are seen as directed and such that all arcs are oriented away from the root. Unless stated otherwise, all trees
in this paper are unrooted.

 Given an unrooted phylogenetic tree $T$ and a subset $Y \subseteq \L(T)$, we denote by $T |Y$ the tree obtained from the minimal subgraph of $T$ connecting $Y$ when contracting degree-2 vertices.
A \emph{quadset} is a set of four labels.
For a quadset $\{a,b,c,d\}$,  there \new{are three non-isomorphic}\footnote{\new{Isomorphism preserving labels.}}  unrooted binary trees, called \emph{quartets}, which are denoted respectively by $ab|cd$, $ac|bd$, and $ad|bc$, depending on how the central edge splits the four labels. We say that an unrooted tree $T$ \emph{displays} the quartet $ab|cd$ if $T|\{a,b,c,d\}$ is  $ab|cd$. 
We denote the set of  quartets that an unrooted tree $T$ displays by $Q(T)$.  Note that if $T$ is binary, 
then $|Q(T)| = {|\L(T)| \choose 4}$. 
A set of quartets $Q$ on a set $L$ is said to be \emph{complete} if for each quadset $\{a,b,c,d\} \subseteq L$, there is in $Q$ \emph{exactly} one quartet among $ab|cd$, $ac|bd$, and $bc|ad$.

We are now ready to state our optimization problem. 
The \textsc{Weighted Quartet Consensus} problem asks for a tree that has as many quartets as possible in common with a given set of trees on the same set of labels $\leafset$:\\

\noindent \textbf{\textsc{Weighted Quartet Consensus} (\maxwqtet) problem} \\
\noindent {\bf Input}: a set of  
unrooted trees $\T = \{T_1, \ldots, T_k\}$ such that 
$\L(T_1) = \ldots = \L(T_k) = \leafset$.\\
\noindent {\bf Output}: a binary unrooted tree $M$ with $\L(M) = \leafset$ that maximizes 
$\sum_{T \in \T} |Q(M) \cap Q(T)|$. \\

\new{The problem is weighted as each quartet on $\leafset$ is weighted by frequency in $\T$, see below. }

In this paper we will focus on the particular case where the input trees are all binary. In fact, proving the problem NP-hard for this restricted case implies NP-hardness of the general problem.
Note however that relaxing the requirement of the output $M$ to be binary leads to a 
different problem, as one needs to define how unresolved quartets in $M$ are weighted.

In the remainder of the article, we will sometimes consider 
multi-sets of quartets,
that are sets in which the same quartet can appear multiple times.
Denote by $f_{\Q}(q)$ the number of times 
that a quartet $q$ appears in a multi-set $\Q$ (we may write $f(q)$ if $\Q$ is unambiguous).
We say that a tree $T$ \emph{contains} $k$ quartets of $\Q$ if 
there are distinct quartets $q_1, \ldots, q_p \in Q(T)$ such that 
$\sum_{i = 1}^p f(q_i) = k$.
The \textsc{Weighted Quartet Consensus} problem can be rephrased as follows: given 
trees $T_1, \ldots, T_k$, finding a tree $M$ that
contains a maximum number of quartets from $Q(T_1) \uplus Q(T_2) \uplus \ldots \uplus Q(T_k)$, 
\newml{where $\uplus$ denotes multiset union}.
We will implicitly work with the decision version of \maxwqtet, i.e.\ for a given integer $q$, 
is there a consensus tree $M$ containing at least $q$ quartets from $Q(T_1) \uplus Q(T_2) \uplus \ldots \uplus Q(T_k)$?

Given a quadset $\{a,b,c,d\}$, we say that $ab|cd$ is \emph{dominant} 
(w.r.t. $f$) if $f(ab|cd) \geq f(ac|bd)$ 
and $f(ab|cd) \geq f(ad|bc)$.  We say that $ab|cd$ is \emph{strictly dominant} 
if both inequalities are strict.

\section{NP-hardness of the \textsc{Weighted Quartet Consensus} problem\label{sec:hardness}}

In this section, we present a reduction from the \textsc{Cyclic ordering} problem.
\newml{This problem has been used in phylogenetics before in~\cite{jansson2001complexity} in the context 
of inferring rooted binary trees from rooted triplets that are not required to originate from 
a set of trees on the same leaf set.}

But beforehand, we need some additional notation.
A \emph{caterpillar} is an unrooted binary tree obtained by taking a path $P = p_1p_2 \ldots p_{r}$, 
then adding a leaf $\l_i$ adjacent to $p_i$ for each $1 \leq i \leq r$, 
then adding another leaf $\l_1'$ adjacent to $p_1$ and a leaf $\l_r'$ adjacent to $p_r$.  
The two leaves $\l_1'$ and $\l_r'$ inserted last are called the \emph{ends} of the caterpillar.
An \emph{augmented caterpillar} $T$ is a binary tree obtained by
taking a caterpillar, then replacing each leaf by a binary rooted tree (replacing a leaf $\l$ by a tree $T'$ means replacing $\l$ by $r(T')$).
If $T_1, T_2$ are the two trees replacing the ends of the caterpillar, 
then $T$ is called a $(T_1, T_2)$-augmented caterpillar.
Note that every binary tree is a $(T_1, T_2)$-augmented caterpillar for some $T_1, T_2$.
Let $T$ be a caterpillar with leaves $(\l_1,\l_2, \dots, \l_k)$ taken in the order in which we encounter them on the path between the two ends $l_1$ and $l_k$ (more precisely, traverse the $\l_1 - \l_k$ path, and each time an internal node is encountered, append its adjacent leaves to the sequence),  
and let $T_1, \ldots, T_k$ be rooted binary trees.
We denote by $(T_1 | T_2 | \ldots | T_k)$ 
the $(T_1, T_k)$-augmented 
caterpillar obtained by replacing \new{$\l_i$} 
by $r(T_i)$, $1 \leq i \leq k$.
This notation gives rise to a natural ordering of its subtrees, and we say that $T_i < T_j$  if $i < j$ (i.e.\ $T_i$ appears before $T_j$ in the given ordering of the caterpillar subtrees).  Note that by reversing such an ordering, we obtain the same unrooted tree.  However, in the proofs the given ordering will be important.
Also, since $T_1, T_2$, and $T_{k - 1}, T_k$ are interchangeable, 
we will simply say that these two pairs are incomparable.
If each $T_i$ consists of a single leaf $\l_i$ for $2 \leq i \leq k - 1$, then we 
may write $(T_1 | \l_2 | \ldots | \l_{k - 1}  | T_k)$, and 
$\l_i < \l_j$ in $T$ to indicate that $\l_i$ appears before $\l_j$ in the ordering.  



We are now ready to describe the \textsc{Cyclic Ordering} problem.
Let $L = (s_1, \ldots, s_n)$ be a linear ordering of a set $S$ of at least 3 elements, and let $(a, b, c)$ be an 
ordered triple, with $a,b,c \in S$.  We say that $L$ \emph{satisfies} $(a, b, c)$ if 
one of the following holds in $L$: $a < b < c, b < c < a$ or $c < a < b$.
If $\C$ is a set of ordered triples we say that $L$ satisfies $\C$ if 
it satisfies every element of $\C$.
Intuitively speaking, $L$ satisfies $(a, b, c)$ when, by turning $L$  into a directed cyclic order by attaching $s_n$
to $s_1$, one can go from $a$ to $b$, then to $c$ and then to $a$.
This leads to the following problem definition: \\

\noindent \textbf{\textsc{Cyclic Ordering} problem} \\
\noindent {\bf Input}: A set $S$ of $n$ elements and a set $\C$ of $m$ ordered triples $(a, b, c)$ of distinct members of $S$.\\
\noindent {\bf Question}: Does there exist a linear ordering $L = (s_1, \ldots, s_n)$ of $S$ that satisfies $\C$? \\

The \textsc{Cyclic Ordering} problem is NP-hard~\cite{galil1977cyclic}.  
In the rest of this section, we let $S$ and $\C$ be the input set and triples, 
respectively, of a \textsc{Cyclic Ordering} problem instance.  
We denote $n = |S|$ and $m = |\C|$.  We shall use the following simple yet useful characterization.

\begin{lemma}\label{lem:triples-to-pairs}
A linear ordering $L$ of $S$ satisfies $\C$ if and only if  
for each $(a, b, c) \in \C$, exactly two of the following relations hold in $L$:
$a < b, b < c, c < a$.
\end{lemma}

\begin{proof}
$( \Rightarrow )$: let $L$ be a \new{linear} ordering satisfying $\C$, and let $(a, b, c) \in \C$.  Thus in $L$, one of $a < b < c, b < c < a$ or $c < a < b$ holds.  
It is straightforward to verify that in each case, exactly two of $a < b, b < c, c < a$
hold.\\
$( \Leftarrow )$: suppose that $L$ does not satisfy $\C$.
Then there is some $(a, b, c) \in \C$ such that 
one of $a < c < b, b < a < c$ or $c < b < a$ does not hold.  Again, one can easily verify that each of these cases satisfies only one of $a < b, b < c$ and $c < a$.
\end{proof}

Now, from $S$ and $\C$ we construct a set of unrooted binary trees $\T$ on the same set of labels 
(we will omit the straightforward verification that this construction can be carried out in polynomial time).
First let $W$ and $Z$ be two rooted binary trees each on $(nm)^{100}$ leaves
(the topology is arbitrary, and the $100$ exponent could be optimized).  Our trees are defined on the leaf set
\newml{$\leafset = S \cup \L(W) \cup \L(Z)$ (note that $S, \L(W), \L(Z)$ are disjoint).} 
Let $C \in \C$ with $C = (a,b,c)$.
We construct $6$ trees from $C$, that is\new{,} $3$ pairs of trees:

\begin{itemize}
\item
The ``$a < b$'' trees:
let $(s_1, \ldots, s_{n - 2})$ be an arbitrary ordering of 
$S \setminus \{a, b\}$.  Then we build the trees 
$\treeord{C}{ab} = (W | a | b | s_1 | s_2 | \ldots | s_{n - 2} | Z)$ 
and 
$\treerev{C}{ab} = (W |  s_{n - 2} | s_{n - 3} | \ldots | s_1 | a | b | Z)$.

\item
The ``$b < c$'' trees:
let $(\hat{s}_1, \ldots, \hat{s}_{n - 2})$ be an arbitrary ordering of 
$S \setminus \{b, c\}$.  Then we build the trees 
$\treeord{C}{bc} = (W | b | c | \hat{s}_1 | \hat{s}_2 | \ldots | \hat{s}_{n - 2} | Z)$ 
and 
$\treerev{C}{bc} = (W |  \hat{s}_{n - 2} | \hat{s}_{n - 3} | \ldots | \hat{s}_1 | b | c | Z)$.

\item
The ``$c < a$'' trees:
let $(\bar{s}_1, \ldots, \bar{s}_{n - 2})$ be an arbitrary ordering of 
$S \setminus \{c, a\}$.  Then we build the trees 
$\treeord{C}{ca} = (W | c | a | \bar{s}_1 | \bar{s}_2 | \ldots | \bar{s}_{n - 2} | Z)$ 
and 
$\treerev{C}{ca} = (W |  \bar{s}_{n - 2} | \bar{s}_{n - 3} | \ldots | \bar{s}_1 | c | a | Z)$.

\end{itemize}

Denote by $\T(C)$ the set of $6$ constructed trees for $C \in \C$.
In this section, the input for our \textsc{Weighted Quartet Consensus} instance constructed from 
$S$ and $\C$ is $\T = \bigcup_{C \in \C}\T(C)$.
Note that $|\T| = 6m$.

Observe that each tree of $\T(C)$ is a $(W, Z)$-augmented caterpillar.
Moreover, note that the \new{majority of ordered pairs are ``balanced''} in the pairs of constructed trees: 
Let $a, b \in S$ and $x, y \in S \setminus \{a, b\}$, and let $\{\treeord{C}{ab}, \treerev{C}{ab}\}$ be an ``$a < b$'' tree-pair.  Then we have $x < y$ in $\treeord{C}{ab}$ if and only if $y < x$ in $\treerev{C}{ab}$.  Similarly for any $x \in S \setminus \{a, b\}$, 
$a < x, b < x$ in $\treeord{C}{ab}$ but $x < a, x < b$ in $\treerev{C}{ab}$.
Only $a < b$ holds in both trees.

Let $T \in \T$, and let $B(T)$ denote the set of quartets of $T$ that have 
at least two members of $\L(W)$, or at least two members of $\L(Z)$.
Thus $B(T)$ consists in all the quartets of the form $w_1w_2|xy$ and 
$z_1z_2|xy$ of $T$, where $w_1, w_2 \in \L(W), z_1, z_2 \in \L(Z)$ and
$x, y \in \leafset$ (note that no quartet of $B(T)$ has the form $w_1x|yw_2$ for $x,y \notin \L(W)$, nor the form $z_1x|yz_2$ for $x,y \notin \L(Z)$).
Note that for any tree $T' \in \T$, $B(T) = B(T')$.
Let $K := 6m|B(T)|$ be the total number of such quartets in $\T$, i.e.\
$K$ is the size of $\biguplus_{T \in \T} B(T)$.
We observe the following:

\begin{remark}
Any $(W, Z)$-augmented caterpillar on $\leafset$ contains
\newml{the $K$ quartets $\biguplus_{T \in \T} B(T)$}.\label{rem:k}
\end{remark}

\noindent
Now, denote $\O := 3m|W||Z|\left(  \binom{n - 2}{2} + 2(n - 2) \right)$.
Let $T \in \T$ and suppose that $T$ is an ``$a<b$'' tree, for some $a,b \in S$.  
For $w \in \L(W)$ and $z \in \L(Z)$, $x,y \in S$,  
a quartet $wx|yz$ \new{displayed by} $T$ is called an \emph{out-quartet} of $T$ if 
$\{x, y\} \neq \{a, b\}$, and an \emph{in-quartet} of $T$ if 
$x = a$ and $y = b$ (note that $x = b$ and $y = a$ is not possible, by construction).  Let $out(T)$ and $in(T)$ denote the set of out-quartets 
and in-quartets, respectively, of $T$.    
Note that each tree $T$ has $|W||Z|$ in-quartets and 
$|W||Z|\left(  \binom{n - 2}{2} + 2(n - 2) \right)$ out-quartets
(because there are $\binom{n - 2}{2} + 2(n - 2)$ ways 
to choose $\{x, y\} \neq \{a, b\}$).  Thus $\O$ is half the total number of out-quartets.  

\begin{lemma}\label{lem:out-quartets}
Any weighted quartet consensus tree $M$ for $\T$ contains at most $\O$ quartets from 
$\biguplus_{T \in \T}out(T)$.  Moreover, if $M$ is a 
$(W, Z)$-augmented caterpillar $(W|s_1|\ldots|s_n|Z)$, where $S = \{s_1, \ldots, s_n\}$, then $M$ contains exactly $\O$ quartets 
from $\biguplus_{T \in \T}out(T)$.
\end{lemma}

\begin{proof}
Let $w \in \L(W)$ and $z \in \L(Z)$.
Let $\{\treeord{C}{ab}, \treerev{C}{ab}\}$ be an ``$a<b$'' tree-pair of $\T$, for some $a,b \in S$, and let $x,y \in S$ such that $\{x, y\} \neq \{a, b\}$.
Because $x < y$ in $\treeord{C}{ab}$ if and only if 
$y < x$ in $\treerev{C}{ab}$, we get that the out-quartet
$wx|yz$ is in $\treeord{C}{ab}$ if and only if $wy|xz$ is in $\treerev{C}{ab}$.
Since $M$ can only contain one of the two quartets, 
it follows that $M$ can contain at most half of the quartets from $out(\treeord{C}{ab}) \uplus out(\treerev{C}{ab})$.
Thus $M$ contains at most half the quartets from $\biguplus_{T \in \T} out(T)$, 
which is $3m |W||Z|\left(  \binom{n - 2}{2} + 2(n - 2) \right) = \O$. 
As for the second assertion, 
if $M = (W|s_1| \ldots |s_n|Z)$ then $M$ contains one of
$wx|yz$ or $wy|xz$ for each $x, y \in S$.  Thus if $M$ does not contain 
the \new{out-quartet} $wx|yz$ from $\treeord{C}{ab}$, then it contains the \new{out-quartet} $wy|xz$
from $\treerev{C}{ab}$.  We deduce that $M$ contains at least half the quartets
from $out(\treeord{C}{ab}) \uplus out(\treerev{C}{ab})$, and thus half the quartets from 
$\biguplus_{T \in \T} out(T)$.
\end{proof}

What follows is a key Lemma.  The proof is not so straightforward and can be found in Appendix \ref{app:proof-lem:all-wz-augm}.

\begin{lemma}\label{lem:all-wz-augm}
Any optimal consensus tree for $\T$
is a $(W, Z)$-augmented caterpillar.
\end{lemma}

We finally arrive at our main result.

\begin{theorem}
The \textsc{Weighted Quartet Consensus} problem is NP-hard.
\end{theorem}

\begin{proof}
We show that 
there exists a linear ordering of $S$ satisfying $\C$ if and only if 
there exists a weighted quartet consensus tree $M$ for $\T$ that contains 
at least
$K + \O + 4m|W||Z|$
quartets from $\biguplus_{T \in \T}Q(T)$.
For the rest of the proof, 
we let $w \in \L(W)$ and $z \in \L(Z)$ be arbitrary leaves of $W$ and $Z$, 
respectively.

$( \Rightarrow )$: let $L = (s_1, s_2, \ldots, s_n)$ be a linear ordering of $S$ satisfying $\C$. Then we claim that the weighted quartet consensus tree $M = (W | s_1 | s_2 | \ldots | s_n | Z)$ contains the desired number of quartets.
Since $M$ is a  $(W, Z)$-augmented caterpillar, $M$ contains $K$ quartets of $\T$ that have 
two or more elements from $\L(W)$, or two or more elements from $\L(Z)$, see remark on page \pageref{rem:k}.
Moreover by Lemma~\ref{lem:out-quartets}, $M$ contains $\O$ 
\newml{quartets} from $\biguplus_{T \in \T} out(T)$.
As for the in-quartets, let $(a,b,c) \in \C$ and let $\T((a,b,c))$ be the set of $6$ trees corresponding to $(a,b,c)$.  
By Lemma~\ref{lem:triples-to-pairs}, $L$ satisfies two of the \new{relations}
 $a<b, b< c, c< a$ .  This implies that $M$ has exactly two of the 
following quartets: $wa|bz, wb|cz, wc|az$.  Since, for every $w \in \L(W)$ and $z \in \L(Z)$, 
each of these three quartets appears as an in-quartet in exactly two trees of $\T((a,b,c))$
\newml{(e.g. $wa|bz$ is an in-quartet of $\treeord{(a,b,c)}{ab}$ and $\treerev{(a,b,c)}{ab}$)}, it follows that 
$M$ contains $4|W||Z|$ quartets of $\biguplus_{T \in \T((a,b,c))} in(T)$. 
As this holds for every $(a,b,c) \in C$, $M$ contains $4m|W||Z|$ quartets of $\biguplus_{T \in \T} in(T)$.
Summing up, we get that $M$ has at least $K + \O + 4m|W||Z|$ quartets
from $\T$.

$( \Leftarrow )$: suppose that no linear ordering of $S$ satisfies $\C$.
Let $M$ be an optimal  consensus tree for $\T$.
\newml{By Lemma~\ref{lem:all-wz-augm}, we may assume that $M$ is 
 a $(W, Z)$-augmented caterpillar}.
We bound the number of quartets of $\T$ that can be contained in $M$.  

First, \newml{by Lemma~\ref{lem:all-wz-augm}}, $M$ contains $K$ quartets of $\T$ that have at least two elements of $\L(W)$ 
or at least two elements of $\L(Z)$.  As for the quartets with one or zero elements from $\L(W) \cup \L(Z)$, in any tree $T \in \T$ there are at most 
$(|W| + |Z|)n^3$ quartets of the form $wa|bc$ or $za|bc$ with $a,b,c \in S$, 
and at most $n^4$ quartets of the form $ab|cd$ with  $a,b,c,d \in S$.
Thus $M$ 
contains at most $6m((|W| + |Z|)n^3 + n^4) < (|W| + |Z|)mn^5$ 
quartets of $\T$ that are of the form $wa|bc, za|bc$ or $ab|cd$ 
with $a,b,c \in S$ (the inequality holds because $n \geq 3$ and 
$|W|=|Z|=(nm)^{100}$).  
Also, by Lemma~\ref{lem:out-quartets}, $M$ contains at most $\O$ 
quartets from $\biguplus_{T \in \T} out(T)$.  It remains to count the in-quartets.



Let $(a,b,c) \in \C$. The following in-quartets appear, each twice, in 
$\T((a,b,c))$:  $wa|bz$, $wb|cz$, $wc|az$.  It is easy to check that these
three quartets are incompatible, 
\newml{i.e. no tree can contain all three of them,} 
and hence $M$ can have at most two of them.
We deduce that there must be at least two trees $T, \overleftarrow{T}$ of $\T((a,b,c))$ such that 
$M$ contains no quartet from $in(T) \uplus in(\overleftarrow{T})$.
Therefore $M$ contains at most $4|W||Z|$ quartets from 
$\biguplus_{T \in \T((a,b,c))} in(T)$,
and thus at most $4m|W||Z|$ quartets from $\biguplus_{T \in \T}in(T)$ 
\newml{assuming that the $4|W||Z|$ bound is attained for every $(a,b,c) \in \C$.  
We will however show that there must be some $(a,b,c) \in \C$ 
such that $M$ contains only $2|W||Z|$ of the quartets in $\biguplus_{T \in \T((a,b,c))} in(T)$}.

\newml{Now, since $M$ is 
 a $(W, Z)$-augmented caterpillar, we write
$M = (W | T_1 | T_2 | \ldots | T_k | Z)$.}
For some $a \in S$, let $T(a)$ be the tree of $\{T_1, \ldots, T_k\}$
that contains $a$ as a leaf.
Then a quartet $wa|bz$ is in $Q(M)$ if and only if 
$T(a) < T(b)$.  Let $L$ be a linear ordering of $S$ such that 
$T(a) < T(b) \Rightarrow a < b$ in $L$.
Since no linear ordering of $S$ can satisfy $\C$, by Lemma~\ref{lem:triples-to-pairs} there must 
be some $(a,b,c) \in \C$ such that 
only one of $a<b, b<c, c<a$ holds in $L$.
This also means that at most one of $T(a) < T(b), T(b) < T(c), T(c) < T(a)$
holds (because $\neg (a < b) \Rightarrow \neg (T(a) < T(b))$).
Thus $M$ has at most one of the $wa|bz, wb|cz, wc|az$ quartets.
It follows that there are at least $2|W||Z|$ quartets from 
$\biguplus_{T \in \T((a,b,c))} in(T)$ that $M$ does not contain.
Therefore $M$ contains at most 
$4m|W||Z| - 2|W||Z|$ quartets of $\biguplus_{T \in \T} in(T)$.

In total, the number of quartets that $M$ contains from the input is bounded by $K + \O + (|W| + |Z|)mn^5 + (4m - 2)|W||Z| < K + \O + 4m|W||Z|$, by our choice of $|W|$ and $|Z|$.
\end{proof}


The implications of these results for the Weighted Triplet Consensus (WTC) problem are presented in Appendix \ref{sec:triplets}.
The same techniques can be used to show that WTC is NP-hard.

\section{The (non)-structure of \maxwqtet \label{sec:nonStructure}}

\newml{In the rest of this paper,  we aim at designing algorithms building on the fact that the weight of each quartet is not arbitrary, and is rather based on a set of input trees on the same leaf set.}
When designing optimized algorithms for a problem, understanding the relationship between the input and the optimal solution(s) can be of great help. 
In phylogenetics, several problems  are harder 
in the supertree setting, i.e.\ when the input trees do not all contain the same species, than in the consensus setting \newml{as in the \maxwqtet~problem}. An example is the problem of finding an unrooted phylogenetic tree containing \new{as minors} a set of unrooted phylogenetic trees -- the compatibility problem -- \new{which} is 
NP-hard in the supertree setting \cite{steel1992complexity} and polynomially solvable in the consensus setting \cite{aho1981inferring}.
Despite the NP-hardness of \maxwqtet, there may exist some properties inherent to the consensus setting that are useful for devising efficient FPT algorithm, or for establishing lower bounds on the value of an optimal solution in order to develop approximation algorithms.

\newml{In attempt to establish useful properties of the weights of quartets 
in the consensus setting, } 
we initially conjectured that the following relationships between the input trees and the solution(s)  hold.  Despite being seemingly reasonable, we prove all these conjectures false. 

\begin{enumerate}
\item let $D$ be the set of strictly dominant quartets of the input multiset $\Q$, i.e.\ the quartets $ab|cd$ such that $f(ab|cd) >  f(ac|bd)$ and $f(ab|cd) >  f(ad|bc)$.  Then there is a constant $\alpha > 0$ such that there exists an optimal solution 
containing at least $\alpha|D|$ quartets of $D$;

\item if a quartet $ab|cd$ has a higher weight \new{than} the sum of the other quartets on the same quadset, i.e.\  $f(ab|cd) >  f(ac|bd) + f(ad|bc)$, then 
\newml{some optimal solution contains $ab|cd$};

\item more generally, there exists $\beta > 0$ such that if a quartet $ab|cd$ is in a fraction $\beta$ of the input trees, then $ab|cd$ must be in some optimal solution.  In particular, if $ab|cd$ is in \emph{every} input tree, 
\newml{then there is some optimal solution that contains $ab|cd$;}

\item if a quartet $ab|cd$ is in no input tree, 
then no optimal solution contains $ab|cd$.

\item call $ab|cd$
a \emph{strictly least-frequent quartet} if $f(ab|cd) < f(ac|bd)$ and $f(ab|cd) < f(ad|bc)$.
\newml{Suppose that there exists a tree $T^*$ on leaf set $\leafset$ that contains no strictly least-frequent quartet,
and choose such a $T^*$ that contains a maximum number of quartets from the input.  
Then $T^*$ is an optimal solution for \maxwqtet.}
\end{enumerate}

Unfortunately,  we answered negatively to all conjectures, see Appendix \ref{sec:nonStructureProofs}.  

\section{Approximatibility of \maxwqtet\label{sec:approx}}

In this section, we show that \maxwqtet~admits a factor $1/2$ approximation algorithm that runs in polynomial time.
Hereafter, the input set of trees is $\T = \{T_1, \ldots, T_k\}$ and we denote $\Q = Q(T_1) \uplus \ldots \uplus Q(T_k)$.
We say that a minimization (resp. maximization) problem $P$ can be approximated within a factor $\alpha > 1$ (resp. $\beta < 1$) if there is an algorithm that,
for every instance $I$ of $P$, 
runs in polynomial time and outputs a solution of value $APP(I)$ such that 
$APP(I) \leq \alpha OPT(I)$ (resp. $APP(I) \geq \beta OPT(I)$), where $OPT(I)$ is the optimal value of $I$.

As mentioned before, the Complete Maximum Quartet Compatibility \new{(}CMQC\new{)}  problem admits a PTAS, though it can only be applied to the
\maxwqtet~problem when the number of input trees is constant.  There does not seem to exist an easy extension of the 
PTAS algorithm for the case of an unbounded number of trees, 
which makes \maxwqtet~seem ``harder'' than CMQC.
Nevertheless, we give a simple factor $1/2$ approximation algorithm, which is better than the (randomized) factor $1/3$ approximation, the best known so far, for
the general Maximum Quartet Consistency problem in which the given quartet set is not necessarily complete.
\newml{We borrow ideas from~\cite{byrka2010new} to show that this can
be achieved by taking the best solution from either a $1/3$ approximation to \maxwqtet, 
or a factor $2$ approximation to WMQI, the minimization version of \maxwqtet~(see below).}
For two unrooted binary trees $T_1, T_2$ 
on leaf set $\leafset$, denote $d_Q(T_1, T_2) = |Q(T_1) \setminus Q(T_2)|$.  The WMQI problem is defined as follows:\\

\noindent \textbf{\textsc{Weighted Minimum Quartet Inconsistency} (WMQI) problem} \\
\noindent {\bf Input}: a set of 
unrooted trees $\T = \{T_1, \ldots, T_k\}$ such that 
$\L(T_1) = \ldots = \L(T_k) = \leafset$.\\
\noindent {\bf Output}: a tree $M$ with $\L(M) = \leafset$ that minimizes 
$\sum_{T \in \T} d_Q(M, T)$. \\

Note that the WMQI problem is equivalent to finding a minimum (in the multiset sense) number of quartets 
to discard from $\Q$ so that it is compatible.

It is not hard to show that $d_Q$ is a metric.  In particular, 
$d_Q$ satisfies the triangle inequality, i.e.\ for any $3$ trees
$T_1, T_2, T_3$ on the same leaf set, 
$d_Q(T_1, T_3) \leq d_Q(T_1, T_2) + d_Q(T_2, T_3)$.
This leads to a factor $2$ approximation algorithm for WMQI obtained by simply returning the best tree from the input.  Intuitively, the input tree that is the closest to the others cannot be too far from the best solution, which is a median tree in the metric space.  
See~\cite{bansal2011comparing} for details.

\begin{theorem}[\cite{bansal2011comparing}]\label{thm:wmqi-approx}
The following is a factor $2$ approximation algorithm for WMQI:
output the tree $T \in \T$ that minimizes 
$\sum_{T_i \in \T}d_Q(T, T_i)$.
\end{theorem}

In~\cite{bansal2011comparing}, the authors \new{explain} how to compute
$d_Q(T_1, T_2)$ in time $O(n^2)$.  Therefore the factor $2$ approximation can be implemented to run in time $O(k^2 n^2)$, 
by simply computing $d_Q$ between every pair of trees.

Theorem~\ref{thm:wmqi-approx} has a direct implication on the 
approximation guarantees of the ASTRAL algorithm in~\cite{mirarab2014astral},
an implementation of the work from Bryant and Steel~\cite{bryant2001constructing}.  
This algorithm finds, in polynomial time, an optimal solution $M$ for \new{a restricted version of  WMQI where}
every bipartition
 of $M$ is also a bipartition in at least one of the input trees.  The solution $T$ returned by the algorithm of Theorem~\ref{thm:wmqi-approx} above trivially satisfies this condition.
Thus, $M$ is at least as good as $T$, implying the following.

\begin{corollary}
The ASTRAL algorithm is a factor $2$ approximation for WMQI.
\end{corollary}

We do not know whether the factor $2$ is tight for the ASTRAL algorithm - we conjecture that ASTRAL can actually achieve a better approximation ratio.  As shown in the rest of this section, this would have interesting applications for the approximability of \maxwqtet.

Indeed, both \maxwqtet~and WMQI share the same set of optimal solutions, 
but the two problems are not necessarily identical in terms of approximability.  We show however that WMQI can be used to approximate \maxwqtet.  As stated earlier, there is a trivial factor $1/3$ randomized approximation
for \maxwqtet: output a random tree $T$.
Each quartet of $\Q$ has a $1/3$ chance of being contained by 
$T$, and so the expected number of quartets of $\Q$ contained by $T$ 
is $|\Q|/3 = k{n \choose 4}/3$ (here $|\Q|$ denotes the multiset cardinality).  Call this the ``random-tree-algorithm''.
For the sake of having a \emph{deterministic} algorithm, we show the following:

\begin{lemma}\label{lem:can-derandomize}
The ``random-tree-algorithm'' can be derandomized, i.e.\
there is a deterministic algorithm that, in time $O(kn^4 + n^5)$, finds a tree containing at least $|\Q|/3$
quartets from $\Q$. 
\end{lemma}

\begin{proof}
We derandomize the factor $1/3$ algorithm using the standard method of conditional expectation. 
For the simplicity of exposition, we will construct a rooted 
tree $T$ in a top-down manner ($T$ can be unrooted after the construction).
Call a rooted tree $T$ \emph{internally binary} if the only nodes of $T$ that have more than two 
children have only leaves as children.
We start with a fully unresolved internally binary tree $T$ on leaf set $\leafset$ (i.e.\ $T$ consists of a root
whose $n$ children are in bijection with $\leafset$).
We then iteratively split each unresolved node $v$ of $T$ into two subtrees so as to maximize the expected number of quartets that $T$ contains.
We stop when $T$ is a binary tree.

To describe the algorithm more precisely, suppose that $T$ is an internally binary tree on leaf set $\leafset$, and let $v$ be a node of $T$ with more than
$2$ children, say $\{v_1, \ldots, v_m\} \subseteq \leafset$
(if no such $v$ exists, then $T$ is binary and we can stop).
We split $v$ by first removing $\{v_1, \ldots, v_m\}$ from $T$, adding two children $x$ and $y$ to $v$, and reinserting
$v_1, \ldots, v_m$ one after another, each as either a child of $x$ or a child of $y$.  We describe how this choice is made.
Suppose that for $i \geq 1$, 
$v_1, \ldots, v_{i - 1}$ have been reinserted, resulting in the tree $T_{i - 1}$, and that we need to process $v_i$.  
Denote by $T_{i,x}$ 
(resp. $T_{i,y}$)
the tree obtained by inserting $v_i$ as a child of $x$ (resp. of $y$) in $T_{i - 1}$.
We then define a random binary tree $T'_{i,x}$ from $T_{i,x}$ as follows:
for each $v' \in \{v_{i + 1}, \ldots, v_m\}$, reinsert $v'$ as a child of $x$ with probability $1/2$, 
or as a child of $y$ with probability $1/2$.  
Then, replace each non-binary node $w$ with children $\leafset'$ by a rooted binary tree
on leaf set $\leafset'$ chosen uniformly at random.  
We define the random binary tree $T'_{i,y}$ from $T_{i,y}$ using the same process.  

For a random tree $T'$ obtained by the above process and for $q \in \Q$, 
let $I(q, T')$ be an indicator variable for whether $q \in Q(T')$.  That is, $I(q, T') = 1$ if $q \in Q(T')$, and $I(q, T') = 0$ otherwise. Let $I(T') = \sum_{q \in \Q}I(q, T')f_{\Q}(q)$
\footnote{Observe that here, $q \in \Q$ means that there exists at least one occurrence of $q$ in the multiset$\Q$, and so each quartet present in $\Q$ is considered once in the summation, independently of $f_{\Q}(q)$.}.
We seek 
\begin{align*}
\max_{T' \in \{T'_{i,x}, T'_{i,y}\} } \mathbb{E}\left[I(T')\right] &= \max_{T' \in \{T'_{i,x}, T'_{i,y}\}} \mathbb{E}\left[\sum_{q \in \Q}I(q, T')f_{\Q}(q)\right] \\
&=\max_{T' \in \{T'_{i,x}, T'_{i,y}\}} \sum_{q \in \Q} \Pr\left[q \in Q(T')\right]f_{\Q}(q)
\end{align*}

If $T'_{i,x}$ attains this maximum, we insert $v_i$ below $x$, and otherwise we insert $v_i$ below $y$.  After every child $v_i$ of $v$ has been inserted, we process the next non-binary node.  This concludes the algorithm description (we shall detail how to compute $\Pr[q \in Q(T')]$ below).

 If $T$ is an internally binary tree, by a slight abuse of 
 notation define $\mathbb{E}\left[I(T)\right] =\mathbb{E}\left[I(T')\right] $, 
 where $T'$ is the random binary tree obtained by replacing each 
 non-binary node of $T$ on leaf set $\leafset'$ by a random binary 
 tree on leaf set $\leafset'$.  

\begin{claim}\label{lem:max-exp}
Let $T$ be an internally binary tree, and suppose that $\mathbb{E}[I(T)] \geq |\Q|/3$.
Let $v$ be a non-binary node of $T$, and let $T_v$ be the 
tree obtained after splitting $v$ using the above algorithm.
Then $\mathbb{E}[I(T_v)] \geq |\Q|/3$.
\end{claim}

Let $\{v_1, \ldots, v_m\}$ be the children of $v$. To prove the claim, we use induction on the number of processed children of $v$ to show that after each insertion of a child $v_i$, the obtained tree $T_i \in \{T_{i,x}, T_{i,y}\}$ satisfies $\mathbb{E}[I(T'_i)] \geq |\Q|/3$, where $T'_i \in \{T'_{i, x}, T'_{i, y}\}$ is the random tree corresponding to $T_i$ obtained from the above process (i.e.\ reinserting $v_{i + 1}, \ldots, v_m$ randomly under $x$ or $y$, and resolving non-binary nodes randomly).  This proves the statement since $T_m = T_v$
(and thus $\mathbb{E}\left[I(T_v)\right] = \mathbb{E}\left[I(T_m)\right] =\mathbb{E}[I(T^{'}_m)] \geq |\Q|/3$).
As a base case, if $i = 1$ it is easy to see that $T'_{1,x}$ and $T'_{1,y}$ are identical, and that $\mathbb{E}[I(T'_{1, x})] = \mathbb{E}[I(T'_{1, y})] = \mathbb{E}[I(T)] \geq |\Q|/3$.  For $i > 1$, let $T_{i - 1}$ be the tree obtained after inserting $v_{i - 1}$, and suppose without  loss of generality that $T_{i - 1} = T_{i - 1, x}$.
Because, in $T'_{i - 1, x}$, we insert $v_i$ below $x$ or $y$ each with probability $\frac{1}{2}$,  we have

\begin{align*}
\mathbb{E}\left[I(T'_{i - 1, x})\right] &= 
\frac{1}{2} \mathbb{E}\left[I(T'_{i - 1,x}) | v_i \mbox{ is a child of }x\right] + \frac{1}{2} \mathbb{E}\left[I(T'_{i - 1, x}) | v_i \mbox{ is a child of } y\right] \\
&= \frac{1}{2} \left(\mathbb{E}\left[I(T'_{i, x})\right] + \mathbb{E}\left[I(T'_{i, y})\right] \right)
\end{align*}

By induction, we also have $\mathbb{E}[I(T'_{ i - 1, x})] \geq |\Q|/3$.  
Combined with the above equality, we obtain $\frac{1}{2} \left(\mathbb{E}\left[I(T'_{i, x})\right] + \mathbb{E}\left[I(T'_{i, y})\right] \right) \geq |\Q|/3$.
This implies that one of $\mathbb{E}[I(T'_{i, x})]$ 
or $\mathbb{E}[I(T'_{i, y})]$ must be at least $|\Q|/3$.
\popQED

Since the fully unresolved tree $T$ from which we start satisfies $\mathbb{E}[I(T)] \geq |\Q|/3$, Claim~\ref{lem:max-exp}
shows that the algorithm does terminate with a tree containing at least $|\Q|/3$ quartets from $\Q$.  It remains to be show how to compute, when reinserting a node $v_i$, the expectations for $T'_{i,x}$ and $T'_{i,y}$.

In fact, it suffices to be able to compute, for a given quartet $q=ab|cd$, the probability $\Pr[q \in Q(T')]$ for $T' \in \{T'_{i,x}, T'_{i,y}\}$.  
Moreover, if $\Pr[q \in Q(T'_{i,x})] = \Pr[q \in Q(T'_{i,y})]$, then this probability does not contribute to determining which scenario maximizes expectation, and in this case we do not need to consider $q$.  
In particular, if none of $a,b,c,d$ is equal to $v_i$, then $\Pr[q \in Q(T'_{i,x})] = \Pr[q \in Q(T'_{i,y})]$.  Therefore, it is enough 
to consider only quartets in which $v_i$ is included.
We will assume that $v_i = a$.
Moreover, we may assume that two or three of $\{b,c,d\}$ are children of $v$ in $T$ (recall that $v$ is the parent of $v_i$ in $T$), because otherwise the probability that 
$ab|cd$ is in $T'$ is unaffected by whether $a$ is a child of $x$ or a child of $y$.

There are still multiple cases depending on which of $b,c$ and $d$ are children of $v$, and which have been reinserted or have not, but 
this probability can be easily found algorithmically.
Let $U = \{b,c,d\} \cap \{v_{i + 1}, \ldots, v_m\}$, i.e.\ the leaves in $\{b,c,d\}$ that have not been reinserted yet in $T'$.  
We obtain new trees $S'_1, \ldots, S'_h$ by reinserting, in $T'$, the members of $U$ below $x$ or $y$ in every possible way -- there are only
$2^{|U|} \leq 8$ possibilities, so $h \leq 8$.  Then, for 
$ 1 \leq j\leq h$
denote by $S'_j|_q$ the tree $S'_j$ restricted to $\{a,b,c,d\}$
(i.e.\ obtained by removing every leaf not in $\{a,b,c,d\}$, 
then contracting degree $2$ vertices).  Note that $S'_j|_q$ may be non-binary.
We get $\Pr[q \in Q(T')] = \sum_{j = 1}^h \frac{1}{h} \Pr[q \in Q(S'_j|_q)]$.  This is because every leaf in $v_{i + 1}, \ldots, v_m$ other than $b,c,d$ is resinserted
independently from the choice for $b,c,d$, and every non-binary node remaining after the reinsertions is resolved uniformly. The probability $\Pr[q \in Q(S'_j|_q)]$ is straightforward to compute, as only a constant number of cases can occur since $S'_j|_q$ has only $4$ leaves.
We omit the details.

{\bf Time complexity:} we must first preprocess the input in order to compute $f_{\Q}(q)$ for each quartet $q$.  This takes 
time $O(kn^4)$.
As for the computation of $\Pr[q \in Q(T')]$, assume that the lowest common ancestor ($lca$) of two leaves can be found in constant time.  This can be achieved naively by simply storing the $lca$ for each pair of leaves in a table of size $O(n^2)$, and updating the table in time $O(n)$ each time a decision on some $v_i$ is made (this does not hinder the total time complexity of the algorithm, though there are more clever ways to handle dynamic tree $lca$ queries~\cite{cole1999dynamic}).  Then the restrictions $S'_1|_q, \ldots, S'_h|_q$ can be computed in constant time.
It is then straightforward to see that, by the above process, 
$\Pr[q \in Q(T')]$  can be computed in constant time.
Each time a node $v_i$ needs to be reinserted, 
this probability must be computed for the $O(n^3)$ quartets containing $v_i$.  There are $n - 1$ splits to be performed, and each split requires inserting $O(n)$ nodes.  Thus the ``binarization'' process takes total time $O(n^5)$, and altogether the derandomization takes time $O(kn^4 + n^5)$.
\qed
\end{proof}

The above leads to a (deterministic) $1/3$- approximation. This can be used to show the following.
The proof is similar to that of~\cite[Theorem 2]{byrka2010new} and is relegated to  Appendix  \ref{app:thm:approx-from-wmqi}.

\begin{theorem}\label{thm:approx-from-wmqi}
If WMQI can be approximated within a factor $\alpha$, then 
\maxwqtet~can be approximated within a factor $\beta = \alpha/(3\alpha - 2)$.
\end{theorem}

Combined with Theorem~\ref{thm:wmqi-approx} and letting $\alpha = 2$ in Theorem~\ref{thm:approx-from-wmqi}
we get the following.

\begin{corollary}
\maxwqtet~can be approximated within a factor $1/2$
in time $O(k^2n^2 + kn^4 + n^5)$.
\end{corollary}

\section{Fixed-parameter tractability of \maxwqtet\label{FPTAlgos}}

In this section we describe how, based on previous results on the minimum quartet incompatibility problem on complete sets,  \maxwqtet~can be solved in time
$O(4^{d' + k_2' + k_3'}n + n^4)$.  Here $k_2'$ and $k_3'$ are the number of 
quadsets that have $2$ and $3$ dominant quartets, respectively, 
and $d'$ is the number of strictly dominant quartets that we are allowed to reject.
The algorithm makes direct use of  the Gramm-Niedermeyer algorithm~\cite{gramm2001minimum}, 
henceforth called the GN algorithm.

The GN algorithm solves the following problem: 
given a \emph{complete set} of quartets $Q$, find, if it exists, a complete and compatible set of quartets $Q'$ such that 
at most $d$ quartets of $Q'$ are not in the input set $Q$ (i.e.\ $|Q' \setminus Q| \leq d$). 
This is accomplished by repeatedly applying the following theorem:

\begin{theorem}[\cite{gramm2001minimum}]\label{thm:GN}
Let $Q$ be a complete set of quartets.  Then $Q$ is compatible if and only if for each set of five taxa $\{a,b,c,d,e\} \subseteq \leafset$, 
$ab|cd  \in Q$ implies $ab|ce \in Q$ or $ae|cd \in Q$.
\end{theorem}

The idea behind the GN algorithm is as follows:  
find a set of five taxa $\{a,b,c,d,e\}$ that does not satisfy the condition of Theorem~\ref{thm:GN}, then correct the situation by branching into the four possible choices:
\begin{enumerate}
\item remove $ab|cd$ from $Q$ and add $ac|bd$ to $Q$; 
\item remove $ab|cd$ from $Q$ and add $ad|bc$ to $Q$;
\item remove $\{ac|be$,$ae|bc\} \cap Q$  from $Q$ and add $ab|ce$ to $Q$;
\item  remove $\{ac|de$,$ad|ce\} \cap Q$  from $Q$ and add $ae|cd$ to $Q$.  
\end{enumerate}
The quartets added to $Q$ will not be questioned in the following branchings.  
With some optimization, this leads to a $O(4^{d}n + n^4)$ FPT algorithm.

In~\cite{gramm2001minimum}, the authors also note that this algorithm can be extended to sets of quartets $Q$ that contain
ambiguous quadsets, i.e.\ sets $\{a,b,c,d\}$ for which $2$ or $3$ of the possible quartets on $\{a,b,c,d\}$ are in $Q$. 
Suppose there are $k_2$ and $k_3$, respectively, quadsets
that have $2$ and $3$ quartets in $Q$.
The modified algorithm then, in a first phase, branches into the $2^{k_2}3^{k_3}$ ways of choosing one quartet per such quadset, thereby obtaining a complete set of quartets for each possibility.
The GN algorithm is thus applied to the so-obtained complete sets. 
This yields a $O(2^{k_2} \cdot 3^{k_3} \cdot 4^{d}n + n^4)$ algorithm. 

It is not hard to see that this gives an FPT algorithm for \maxwqtet, where the parameter $k_2$ (resp. $k_3$) is the number of quadsets such that $2$ (resp. $3$) possible quartets appear in the input trees, and $d$ is the number of quartets $ab|cd$ that appear in every input tree, and that we are allowed to discharge.
Note however that, in the consensus setting, there is no reason to believe that $k_2$ and $k_3$ are low - in we fact we believe that $k_2 + k_3$ typically takes values in $\Theta(n^4)$.  One reason
is that even the slightest amount of noise on a quadset makes it included in the count of either $k_2$ or $k_3$ (e.g. if $k - 1$ trees agree on $ab|cd$ and only one contains $ac|bd$).

The GN algorithm can, however, be used on a more interesting 
set of parameters.  
Define $k_2'$ (resp. $k_3'$) as the number of quadsets 
that have exactly $2$ (resp. $3$) dominant quartets, and 
let $d'$ be the number of strictly dominant quartets 
that we are allowed to discharge.
It is reasonable to believe that, if each tree of the input is close to the true tree $T^*$, most ``true'' quartets will appear
as strictly dominant in the input, and there should not be too many ambiguous quadsets.
There is a very simple algorithm achieving time $O(4^{d' + k_2' + k_3'}n + n^4)$.
Construct a complete set of quartets $Q$ as follows:
for each quadset $\{a,b,c,d\}$, choose a dominant quartet
on $\{a,b,c,d\}$ and add it to $Q$ (if multiple choices are possible, choose arbitrarily).
Then, run the GN algorithm on $Q$ with the following modification:
each time a quartet $q$ is removed from $Q$ and replaced by another quartet $q'$, decrement either $d', k_2'$ or $k_3'$, depending on whether $q$ belongs to a quadset with $1, 2$ or $3$ dominant quartets.  It follows that if there exists a complete and compatible set of quartets $Q'$ such that at most $d'$ strictly dominant quartets are rejected, then the modified algorithm will find it.  It should be noted however that finding such a set $Q'$ does not guarantee that the corresponding tree is an optimal solution.  Indeed, since quartets are weighted, two solutions $Q'$ and $Q''$ may both reject only $d'$ strictly dominant quartets, yet one has higher weight than the other.
However, the correctness of the algorithm follows from the fact that the GN algorithm finds the set of \emph{every} solution discarding at most $d'$ dominant quartets - and thus it suffices to pick the solution from this set that has optimal weight.

We finally mention that the FPT algorithms published in~\cite{chang2010new} are improved versions of the GN algorithm, can also return every solution and thus can be modified in the same manner.
These yield FPT algorithms that can solve \maxwqtet~in time $O(3.0446^{d' + k_2' + k_3'}n+ n^4)$  and $O(2.0162^{d' + k_2' + k_3'}n^3 + n^5)$.

\section{Conclusion}

In this paper, we have shown that the \maxwqtet~problem is NP-hard, answering a question of~\cite{MirarabThesis} and~\cite{bansal2011comparing}.
In the latter, the authors also propose a variant of the problem in which the output tree $T$ is not required to be binary.
In this case, 
one needs to assign a cost $p$ to the unresolved quartets. 
Our reduction can be extended to show that hardness holds for high enough $p$, but the complexity of the general case remains open.
We have also shown that~\maxwqtet~can be approximated within a factor 1/2.  One open question is whether the problem admits a PTAS as the related CMQC problem.  
The fixed-parameter tractability aspects of WQC also deserve further investigation.
This would require identifying some structural properties that are present in the consensus setting and that can be used for designing practical exact algorithms.
But as we have shown, this might not be an easy task, as many properties which seem reasonable for the consensus setting do not hold.

\bibliography{p23-lafond}

\appendix

\section{Implications for the \textsc{Weighted Triplet  Consensus} problem}\label{sec:triplets}

For each set of three labels $\{a,b,c\}$, there \new{are three non-isomorphic}\footnote{\new{Isomorphism preserving labels and the root node.}} rooted binary trees called \emph{triplets}. They are denoted by $ab|c$, $ac|b$ and $bc|a$, depending on the leaf having the root as father ($c$, $b$ and $a$ respectively). We say that 
\newml{a tree $T$} induces or displays the triplet $ab|c$ if $T|\{a,b,c\}=ab|c$. For a rooted tree $R$, denote by $tr(R)$ the set of  triplets of $R$.

When the consensus is sought for rooted trees,  the objective is
to find a rooted tree $M$ that induces a maximum number of 
 triplets contained in the input trees.  
The \textsc{Weighted Triplet  Consensus} (WTC) is defined as follows.

\noindent \textbf{\textsc{Weighted Triplet  Consensus} (WTC) problem} \\
\noindent {\bf Input}: a set of  
rooted trees $\R = \{R_1, \ldots, R_k\}$ such that 
$\L(R_1) = \ldots = \L(R_k) = \leafset$.\\
\noindent {\bf Output}: a binary rooted tree $M$ with $\L(M) = \leafset$ that maximizes 
$\sum_{R \in \R} |tr(M) \cap tr(R)|$. \\

As in the unrooted problem, other versions of WTC where the input trees  may have missing species or where the weight of a triplet is not defined w.r.t. a set of trees, are known to be NP-hard \cite{bryant1997building}. The WTC problem is conjectured to be NP-hard in \cite{bansal2011comparing}
(we note that a more general version where the output can be non-binary is also conjectured NP-hard).

We give the main idea behind the proof of the hardness of WTC.
Let $\T = \bigcup_{C \in \C}\T(C)$ be the set of unrooted trees constructed in the 
reduction above.  For a tree $T \in \T$, let $e$ be the edge separating $Z$ from the rest of the tree
(i.e.\ by removing $e$ from $T$, one connected component is exactly $Z$).  
Obtain a rooted tree $R$ from $T$ by rooting $T$ at $e$, that is subdivide $e$, thereby creating a degree 
$2$ vertex which is the root of $R$.  The set of rooted trees $\R$ is obtained by applying this rooting
to every $T \in \T$ (the $Z$ subtree could be removed but we keep it here to make the correspondence 
easier to see).

Similarly as above, it can be shown that since every input tree is a rooted $(W,Z)$-caterpillar, 
then any solution must also have this form.  This implies in turn that 
there exists a linear ordering of $S$ satisfying $\C$ if and only if there is a solution $M$ to WTC containing 
every triplet from the input on $2$ or $3$ members of $\L(W)$, 
every triplet containing at least one member of $\L(Z)$, plus at least $4m|W| + 3m|W|\left(  \binom{n - 2}{2} + 2(n - 2) \right)$ 
triplets of the form $wa|b$ with $a,b \in S$.  
\newml{This is obtained by defining the notions of \emph{in-triplets} and \emph{out-triplets}
analogously as in the previous section, but with respect to $W$ only.  That is, in a ``$a < b$'' tree, for $a,b,c, d  \in S, w \in \L(W)$ and $\{c, d\} \neq \{a, b\}$, $wa|b$ would be an in-triplet, whereas $wc|d$ or $wd|c$ would be out-triplets.
One can argue that for a cyclic triple $(a,b,c) \in \C$ and 
the set of trees $\T((a,b,c))$, an optimal consensus tree can contain 
$4|W|$ of the $6|W|$ possible in-triplets, plus at most half of the 
$6m|W|\left(  \binom{n - 2}{2} + 2(n - 2) \right)$ possible out-triplets}.  The arguments are essentially the same as the ones given in the hardness 
proof of \maxwqtet, and so we omit the details.

\begin{theorem}
The \textsc{Weighted Triplet  Consensus} problem is NP-hard.
\end{theorem}


\section{Deferred proofs}

\subsection{Proof of Lemma \ref{lem:all-wz-augm}\label{app:proof-lem:all-wz-augm}}



Despite the Lemma \ref{lem:all-wz-augm} statement being quite intuitive, it requires a surprising amount of care.
We start by a simple proposition that will be needed.

\begin{proposition}\label{prop:setcards}
Let $X, Y$ be two non-empty sets such that $Y \not\subseteq X$.  Then
$|X| \cdot |Y \setminus X| \geq |Y| - 1$.
\end{proposition}

\begin{proof}
Suppose first that $X \cap Y = \emptyset$.  Then clearly $|X||Y \setminus X| = |X||Y| \geq |Y| - 1$.
Suppose otherwise that $X \cap Y \neq \emptyset$, and denote $X' = X \cap Y$.  
Then $|Y \setminus X| = |Y| - |X'|$ and 
since $Y \not\subseteq X$, we must have $|Y| \geq |X'| + 1$.
We also have $|X||Y \setminus X| = |X|(|Y| - |X'|) \geq |X'|(|Y| - |X'|)$; \new{ we claim the latter term} to be at least
$|Y| - 1$.
Let us assume for contradiction that $|X'|(|Y| - |X'|) < |Y| - 1$.  If $|X'| = 1$, this is clearly impossible,
so assume $|X'| > 1$.
Then we get $|X'||Y|  - |Y| <  |X'|^2 -1 $ leading to $|Y| < \frac{|X'|^2 - 1}{|X'| - 1} = |X'| + 1$,
contradicting $|Y| \geq |X'| + 1$.
\end{proof}

Before proceeding, 
we must introduce the notion of a rooted subtree of a binary unrooted tree $T$.
Note that by removing an edge $e=\{u,v\}$ of $T$, 
we obtain two disjoint rooted subtrees $T_1$ and $T_2$, respectively rooted at $u$ and $v$.  Call $T'$ a  \emph{rooted subtree} of $T$ if $T'$ is a rooted tree that can be obtained by removing an edge of $T$. 
For $X \subset \L(T)$, a \emph{rooted subtree for $X$} is a rooted subtree $T'$ of $T$ 
such that $X \subseteq \L(T')$.  
We denote by $T[X]$ the rooted subtree for $X$ that contains a minimum number of leaves
(if there are multiple choices, choose $T[X]$ arbitrarily among the possible choices).  Note that 
$T[X]$ may contain leaves other than $X$.

We now prove that any optimal solution to $\T$ as constructed in our reduction must be a $(W, Z)$-augmented caterpillar.
Suppose that $M$ is an optimal solution for $\T$, and that $M$ is not
a $(W, Z)$-augmented caterpillar.  Denote $M_W = M[\L(W)]$ and $M_Z = M[\L(Z)]$.  If $M$  
is a $(W', Z')$-augmented caterpillar  $(W'|T_1|\ldots|T_k|Z')$
for some trees $W', Z'$ with $\L(W') = \L(W)$ and $\L(Z') = \L(Z)$, it is not hard to see that $M' = (W|T_1|\ldots|T_k|Z)$ is a better solution than $M$, a contradiction. Thus,  $M$ is not such a caterpillar, and this implies that either $\L(M_W) \neq \L(W)$ or $\L(M_Z) \neq \L(Z)$ (or both).  That is, the rooted subtrees containing $\L(W)$ and/or $\L(Z)$ have ``outsider'' leaves.  Suppose first that $\L(M_W) \neq \L(W)$ holds.  
Then there exists a node $x$ with children $x_l$ and $x_r$ in $M_W$ such that  
all leaves $X_l$ below $x_l$ are in $\L(W)$ with 
$\L(W) \not\subseteq X_l$ (otherwise $M_W = M[\L(W)]$ would be chosen incorrectly),
 and no leaf $X_r$ below $x_r$ 
belongs to $\L(W)$ (this can be seen by observing that the minimal node $x$ 
of $M_W$ having leaves both in $W$ and not in $W$ has this property).  

We claim that $\L(Z) \not\subseteq X_r$.  Suppose otherwise that $\L(Z) \subseteq X_r$.  
Then $|X_r| \geq |Z|$ and so 
$|M_W| \geq |W| + |Z|$.  However in $M$, by removing the $x x_r$ edge we obtain two rooted trees, one of which
is a rooted subtree for $\L(W)$.  Moreover, this subtree has at most $|W| + |S| < |W| + |Z|$ leaves, 
which contradicts the minimality of $M_W = M[\L(W)]$.  We deduce that $\L(Z)$ is not a subset of $X_r$.

Now, observe that $M$ contains the quartet 
$w_1y|w_2z$ for each $w_1 \in X_l, y \in X_r, w_2 \in \L(W) \setminus X_l, z \in \L(Z) \setminus X_r$.  
There are at least $|X_l||X_r|(|\L(W) \setminus X_l|)(|\L(Z) \setminus X_r|) \geq (|W| - 1)(|Z| - 1)$ such quartets (the inequality is obtained by applying Proposition~\ref{prop:setcards} to $|X_l| \cdot |\L(W) \setminus X_l|$ 
and $|X_r| \cdot |\L(Z) \setminus Z|$).
Moreover, each input tree of $\T$ contains the quartet $w_1w_2|yz$ instead, 
and hence in total in $\T$ there are at least $6m(|W| - 1)(|Z| - 1)$ quartets of the form $w_1w_2|yz$ that $M$ does not contain.
In the same manner, if the case $\L(M_Z) \neq \L(Z)$ holds, then there are at least 
$6m(|W| - 1)(|Z| - 1)$ quartets of the form $z_1z_2|yw$ that $M$ does not contain, where here $z_1, z_2 \in \L(Z), y \notin \L(Z), w \in \L(W)$.

Now, let $\rho(M)$ be the number of quartets that $M$ contains 
 from $\biguplus_{T \in \T}Q(T)$ that have the form $wx|yz$, where 
 $w \in \L(W), z \in \L(Z), x,y \in S$.
Formally, 

$$
\rho(M) = \sum_{\substack{wx|yz \in Q(M)  \\ x, y \in S \\ w \in \L(W) \\ z \in \L(Z)} } f(wx|yz)
$$

where $f(wx|yz)$ denotes the number of trees of $\T$ that contain the $wx|yz$ quartet.
For a given $u \in \L(W) \cup \L(Z)$, let $\rho(M, u)$ denote the number of quartets counted in $\rho(M)$ that contain $u$.  Formally, if $w \in \L(W)$, 
we have 

$$\rho(M, w) = \sum_{\substack{wx|yz \in Q(M)  \\ x, y \in S \\ z \in \L(Z)} } f(wx|yz)$$ 

The definition of $\rho(M, z)$ is the same for $z \in \L(Z)$, except that $z$ gets fixed instead of $w$ in the summation.
Notice that $\rho(M) = \sum_{w \in \L(W)} \rho(M, w) = \sum_{z \in \L(Z)} \rho(M, z)$.  
Let $w^* = \arg \max_{w \in \L(W)}\{\rho(M, w)\}$.  
We obtain an alternative solution $M'$ 
from $M$ in the following manner: remove all leaves of $\L(W) \setminus \{w^*\}$ from $M$, delete the degree $2$ nodes, and replace 
$w^*$ by the $W$ tree.  Note that if $w^*x|yz$ is a quartet of $M$, then 
$wx|yz$ is a quartet of $M'$ for all $w \in \L(W)$, and so 
$\rho(M', w) \geq \rho(M, w)$ for all such $w$ by the choice of $w^*$.
Consequently, $\rho(M') \geq \rho(M)$.
We repeat the same operation on $M'$ for the $Z$ tree and obtain our final tree
$M^*$.  That is, we find 
$z^* = \arg \max_{z \in \L(Z)}\{\rho(M', z)\}$, and replace $z^*$ by the $Z$ tree.  As above, we obtain $\rho(M^*) \geq \rho(M')$. 
Since $M^*$ has $W$ and $Z$ as rooted subtrees, it follows that $M^*$ is a
$(W, Z)$-augmented caterpillar.

We argue that $M^*$ contains more quartets from the input trees than $M$.  First observe that the quartets on which $M$ and $M^*$ differ must contain a member of $\L(W) \cup \L(Z)$, since only these leaves switched position.  The tree $M^*$ contains every quartet of $\biguplus_{T \in \T}Q(T)$ that have at least two members of $\L(W)$, or two members of $\L(Z)$. 
This includes the aforementioned (at least) $6m(|W| - 1)(|Z| - 1)$ quartets of the form 
$w_1w_2|yz$ or $z_1z_2|yw$ that $M$ does not contain.  
As for the quartets that contain one member of $\L(W)$ and one member of $\L(Z)$, $M^*$ contains at least as many such quartets as $M$ since in
$\biguplus_{T \in \T}Q(T)$, these quartets are all of the form $wx|yz$, and we have $\rho(M^*) \geq \rho(M)$.  Finally, each tree of $\T$ has at most $(|W| + |Z|)n^3$ quartets that have exactly one member of $\L(W) \cup \L(Z)$.  Thus at most $6m(|W| + |Z|)n^3$ quartets
of this type are contained by $M$ and not contained by $M^*$, but since this is smaller than $6m(|W| - 1)(|Z| - 1)$ for our choice of $|W|$ and $|Z|$, $M^*$ contains more quartets from the input than $M$.

\subsection{Proofs of Section \ref{sec:nonStructure}\label{sec:nonStructureProofs}}

\begin{figure}
\centering
\includegraphics[scale=0.75]{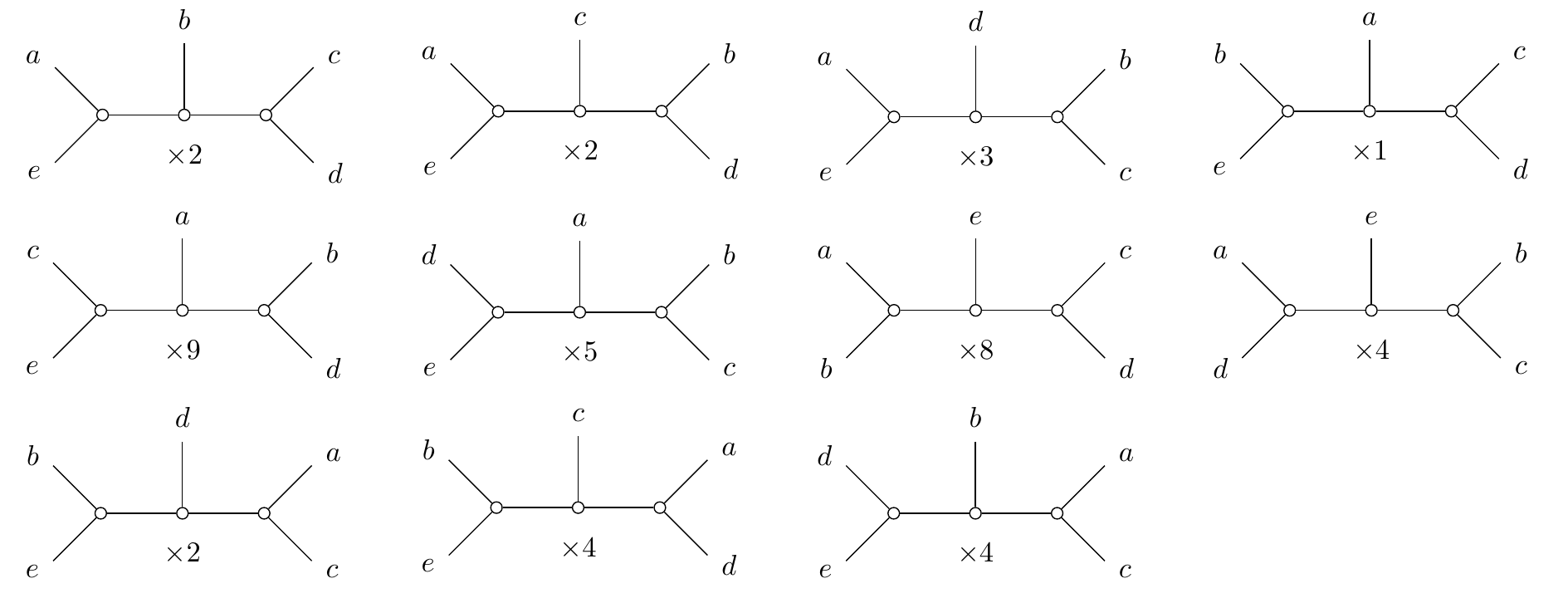}
\caption{An instance of \maxwqtet~such that the optimal solution (the third tree on the first row) contains no strictly dominant quartet.  The numbers correspond to the number of times that each tree appears in the input.}\label{fig:no-dominant}
\end{figure}

\begin{figure}[h]
\centering
\includegraphics[width=0.8\linewidth]{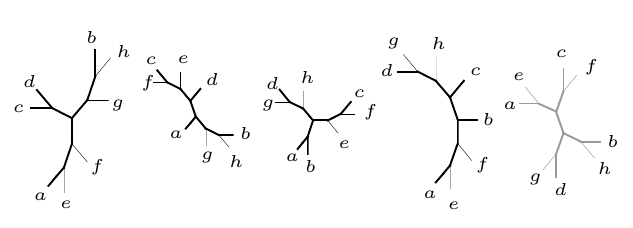}
\caption{The first four trees form an instance of \maxwqtet~in 
which every tree contains $ab|cd$.  The rightmost tree is the unique optimal solution to the \maxwqtet~instance (every possible solution was verified computationally).}\label{fig:no-univqtet}
\end{figure}

\new{C}onjecture 1 is disproved by Theorem~\ref{thm:no-dominant}, and \new{C}onjecture 3 by Theorem~\ref{thm:no-univqtet}, which implies that \new{C}onjectures 2 and 4 are also false; finally \new{C}onjecture 5 is disproved by 
Theorem~\ref{thm:dont-trash-least-frequent}.

\begin{theorem}
There exists an instance of \maxwqtet~such that every optimal solution contains none of the strictly dominant quartets.\label{thm:no-dominant}
\end{theorem}

Figure~\ref{fig:no-dominant} shows an instance of \maxwqtet~demonstrating Theorem~\ref{thm:no-dominant}.
In this instance, for every quadset $S$, there is a strictly dominant quartet appearing $17$ times, whereas the second-most and third-most quartets appear in $16$ and $11$ trees, respectively.  For example, $f(ac|bd) = 17, f(ad|bc) = 16$ and $f(ab|cd) = 11$.
One can check that the best tree is the third one on the top row (the $ae|bc$ with $d$ grafted on the middle branch).  Call this tree $T^*$.  For every quadset $S$, $T^*$ contains 
the second-most frequent quartet on $S$.  The reason why $T^*$ is optimal is that\newml{, in the particular instance of Figure~\ref{fig:no-dominant},} any other tree $T$ that contains a strictly dominant quartet for some quadset $S$ must also contain 
a least frequent quartet on some other quadset $S'$.
Hence, as there are $5$ quadsets, $T$ contains at most $4 \cdot 17 + 11 = 79$ quartets from the input, whereas $T^*$ contains 
$5 \cdot 16 = 80$.
Note that this example consists \new{of} trees on only $5$ leaves.
We do not know if such instances exist for any $n > 5$ leaves.

\begin{theorem}\label{thm:no-univqtet}
There exists an instance of \maxwqtet~such that there is a quartet $q$ that appears in every input tree,
but $q$ is not a quartet of any optimal solution. 
\end{theorem}

Figure~\ref{fig:no-univqtet} shows an instance of \maxwqtet~proving Theorem~\ref{thm:no-univqtet}.
Each input tree contains the $ab|cd$ quartet, whereas the optimal solution, which is unique, does not.   The rightmost tree contains $180$ quartets from the input multiset $\Q$, whereas any other tree has at most $176$.

Finally, we note that the main interest behind \new{C}onjecture 5 is the following: \new{if it holds, in cases where the set $F$ of strictly least-frequent quartets is complete we could tell in polynomial time -- using results of~\cite{bryant2001constructing} -- whether there is a tree $T^*$ that contains no quartet from $F$}.   Conjecture 5 could then lead to interesting \new{ approximations} or FPT algorithms.
However, least-frequent quartets cannot be excluded automatically.

\begin{theorem}\label{thm:dont-trash-least-frequent}
There exists an instance of \maxwqtet~such that every optimal 
solution contains a strictly least-frequent quartet, even if there exists a tree $T^*$ with no such quartet.
\end{theorem}

The instance corresponding to Theorem~\ref{thm:dont-trash-least-frequent} is obtained from the instance shown in Figure~\ref{fig:no-dominant}, by removing all occurrences of the third tree on the top row 
(i.e.\ this tree now appears $0$ times instead of $3$ times).
The second-most \newml{frequent} quartets now appear $13$ times each, and so the tree $T^*$ that contains all these quartets has a total weight of $5 \cdot 13 = 65$.  However, there are trees with a total weight of $75$, which are optimal (for instance, the tree of cardinality $9$ in the \new{f}igure).  Each such tree contains a strictly dominant quartet, and as mentioned before, also a strictly least-frequent quartet.

\subsection{Proof of Theorem~\ref{thm:approx-from-wmqi} \label{app:thm:approx-from-wmqi}}

Let \new{$N := k{n \choose 4}$}, i.e.\ the total number of quartets in $\Q$, 
let $p$ be the maximum number
of quartets that can be preserved from $\Q$ for compatibility, and let $d$ be the minimum number of quartets to discard from $\Q$ in order to attain compatibility (here $p$ and $d$ refer to multiset cardinalities).
Note that $d = N - p$.
We show that taking the best tree between the one obtained from the factor $\alpha$ algorithm for WMQI  and the one obtained from the ``random-tree-algorithm'' achieves a factor $\beta$ for \maxwqtet.  
Suppose first that $p \leq N/(3\beta)$.  \new{By Lemma \ref{lem:can-derandomize},} the ``random-tree-algorithm'' yields a tree containing at least \newml{$|\Q|/3 = N/3$} quartets from $\Q$, and since 
$N/3 = \beta N/(3\beta) \geq \beta p$, it yields a solution to \maxwqtet~within a factor $\beta$
from optimal.  Thus we may assume that $p > N/(3\beta) = N(3\alpha - 2)/(3\alpha)$.
Since we have an $\alpha$ approximation for WMQI,
we may obtain a solution discarding at most $\alpha d = \alpha(N - p)$
quartets.  This solution preserves at least $N - (\alpha(N - p)) = \alpha p + (1 - \alpha)N$ quartets from $\Q$.  We claim that this attains a factor $\beta$ approximation.  Suppose instead that $\alpha p + (1 - \alpha)N < \beta p$.
Then $p < (\alpha - 1)N/(\alpha - \beta)$ which, with a little work, yields $p < N(3\alpha - 2)/(3\alpha)$, contradicting our assumption on $p$.
Thus, the  WMQI approximation preserves at least $\beta p$ quartets.

\end{document}